\pdfoutput=1
\documentclass[a4paper,reprint,twocolumn,notitlepage,aip,nofootinbib]{revtex4-2}

\usepackage[left=1.5cm,right=1.5cm,top=1cm,bottom=1.5cm,includeheadfoot]{geometry}

\usepackage{tgtermes} 
\usepackage{tgheros} 
\usepackage[T1]{fontenc}

\usepackage{amssymb}
\usepackage{amsmath}
\usepackage{amsthm}
\usepackage[dvipsnames]{xcolor}
\usepackage{soul}
\usepackage{float}
\usepackage{multirow}
\usepackage{tikz}
\usetikzlibrary{positioning,arrows}
\usepackage{enumerate}

\PassOptionsToPackage{hyphens}{url} 
\usepackage[breaklinks=true]{hyperref}
\bibpunct{[}{]}{;}{n}{}{} 

\renewcommand\labelenumi{(\roman{enumi})}
\renewcommand\theenumi\labelenumi

\newcommand \Tr {\mathop\mathrm{Tr}}

\newcommand \argmin {\mathop\mathrm{arg\,min}}
\mathchardef\mhyphen="2D 

\newcommand{\R}{{\mathbb R}}
\newcommand{\eps}{\varepsilon}
\renewcommand{\epsilon}{\varepsilon}

\newcommand{\rr}{{\mathbf{r}}}

\newcommand{\rmd}{\,\mathrm{d}}
\newcommand{\rme}{\mathrm{e}}
\newcommand{\proxf}{\Pi_{f}^\eps}
\newcommand{\rhogs}{\rho_\mathrm{gs}}
\newcommand{\subdiff}{\underline\partial}
\newcommand{\superdiff}{\overline\partial}
\newcommand{\id}{\mathrm{id}}
\newcommand{\vxc}{v_\mathrm{xc}}
\newcommand{\vH}{v_\mathrm{H}}
\newcommand{\vKS}{v_\mathrm{KS}}
\newcommand{\Href}{H_0^\mathrm{ref}}

\newcommand{\Fcnl}{\mathcal F}
\newcommand{\Fcnlp}{\mathcal G}

\newcommand{\EE}{E_\Fcnl}
\newcommand{\DHe}{D}
\newcommand{\CS}{\mathrm{CS}}

\newcommand{\dua}[2]{\langle #1, #2 \rangle}
\newcommand{\ext}{\mathrm{ext}}
\newcommand{\KS}{\mathrm{KS}}

\newtheorem{theorem}{Theorem}

\newtheorem{lemma}{Lemma}

\newtheorem{procedure}{Procedure}

\newtheorem{remark}{Remark}

\begin{document}

\title[Density-potential inversion from Moreau--Yosida regularization]{Density-potential inversion from Moreau--Yosida regularization\\\vspace{0.2cm}}

\author{Markus Penz}
\address{Basic Research Community for Physics, Innsbruck, Austria}

\author{Mihály A.\ Csirik}
\address{Department of Computer Science, Oslo Metropolitan University, Norway}
\address{Hylleraas Centre for Quantum Molecular Sciences, Department of Chemistry, University of Oslo, Norway}

\author{Andre Laestadius}
\email{andre.laestadius@oslomet.no}
\address{Department of Computer Science, Oslo Metropolitan University, Norway}
\address{Hylleraas Centre for Quantum Molecular Sciences, Department of Chemistry, University of Oslo, Norway}

\begin{abstract}
\vspace{.2cm}
For a quantum-mechanical many-electron system, given a density, the Zhao--Morrison--Parr method allows to compute the effective potential that yields precisely that density. 
In this work, we demonstrate how this and similar inversion procedures mathematically relate to the Moreau--Yosida regularization of density functionals on Banach spaces. It is shown that these inversion procedures can in fact be understood as a limit process as the regularization parameter approaches zero. 
This sheds new insight on the role of Moreau--Yosida regularization in density-functional theory and allows to systematically improve density-potential inversion. Our results apply to the Kohn--Sham setting with fractional occupation that determines an effective one-body potential that in turn reproduces an interacting density. 
\vspace{1.2cm}
\end{abstract}

\maketitle

\section{Introduction}

In cases where the one-body ground-state density $\rhogs$ of a quantum-mechanical many-electron system is known, the Zhao--Morrison--Parr (ZMP) method \cite{ZP1993,ZMP1994}, among others \cite{vanleeuwen1994exchange,kumar2019universal,Shi2021}, allows to determine the effective potential of a reference system that reproduces exactly that density. 
Such inverse-problem methods can either be used to directly match experimental data \cite{jayatilaka1998wave}, or to gain valuable insight into approximations in density-functional theory (DFT) \cite{vonBarth2004basic,burke2007abc,dreizler2012-book,eschrig2003-book}. 
The resulting density-potential map has been recently studied from a mathematical point of view in great detail \cite{garrigue2021some,garrigue2022building}.
In particular, these methods allow to determine the effective potential, $\vKS$, of a Kohn--Sham (KS) system~\cite{KS} of non-interacting particles and thus also the exchange-correlation potential, $\vxc$. Specialized methods for inversion in KS systems have also been devised \cite{wu2003direct,jensen2018numerical} and the time-dependent case recently received attention as well \cite{nielsen2018numerical}.
The ZMP method itself has been implemented and used on numerous occasions, mainly in the context of KS-DFT \cite{morrison1995-ZMP,tozer1996-ZMP,liu1999-ZMP,harbola2004-ZMP,varsano2014-ZMP,chauhan2017-ZMP,Nam-et-al-2021}.

In this work, we focus on methods for non-relativistic $N$-electron quantum systems.
Here, $v$ will be a one-body potential and $V$ the corresponding lifted $N$-body operator, i.e., $V= \sum_{j=1}^N v(\rr_j)$.  
Now take $\rhogs$ as a ground-state density of some Hamiltonian $H$ and $\Href$ a reference Hamiltonian. A motivating problem for our discussion is as follows: How does the  (effective) potential $v$ needs to be chosen such that $\Href + V$ has exactly the given ground-state density $\rhogs$?
For $H$ modeling an interacting system in an external potential $v_\ext$ and taking $\Href = -\frac{1}{2}\sum_{j=1}^N\nabla_j^2 =:K$, we have $v = \vKS = v_\ext + \vH + \vxc$, i.e., the KS potential that contains the external potential of $H$, the Hartree potential $\vH =\rhogs * |\cdot|^{-1}$, and the exchange-correlation potential. 
This effectively maps the interacting problem to a non-interacting one.
Contrary to that, in the case $H = \Href + V$, our result addresses the direct density-potential inversion.
\vfill\eject

The ZMP method was originally derived using a penalty parameter $\lambda$ with the density constraint appearing in terms of a Coulomb integral. For a fixed value $\lambda>0$ this method defines 
\begin{equation*}
    v^\lambda(\rr) = \lambda \int \frac{\rho^\lambda(\rr')-\rhogs(\rr')}{|\rr-\rr'|}\,\rmd\rr',
\end{equation*}
where $\rhogs$ is a ground-state density of the Hamiltonian $H$. The pair $(\rho^\lambda,v^\lambda)$ 
is determined self-consistently and the procedure is repeated for increasing values of $\lambda$. Finally, careful numerical extrapolation $\lambda\to\infty$ yields a potential $v^\lambda\to v$ that  has the ground-state density $\rhogs$, as required~\cite{ZMP1994}.

In the following, we will demonstrate how the ZMP and related methods can be justified in a mathematical precise manner by relying on the Moreau--Yosida regularization of a ({\it not} necessarily universal) density functional  on a suitably chosen function space.
The possibility of such a connection between the ZMP method and Moreau--Yosida regularization was already mentioned in Ref.~\cite{Kvaal2014}, but has not yet been made explicit.

\subsection{Outline}
In Section~\ref{zmporigsec}, we briefly review the ZMP method for the reader's convenience. Next, in Section~\ref{sec:my}, we discuss the necessary convex-analysis concepts and Moreau--Yosida regularization with applications to DFT in mind. We then formulate our rather general setting in Section~\ref{sec:M-setting}, to which we will apply the density-potential inversion procedure.
Our main results are discussed in Section~\ref{mainsec}, which contains the derivation of the ZMP method with an additional correction term (Theorem~\ref{thm:ZMP}). 
This is achieved by first formulating an abstract density-potential inversion procedure (Procedure~\ref{proc:ZMP} and also Procedure~\ref{proc:DIFF}) by exploiting a property of the derivative of the Moreau--Yosida regularized functional (Theorem~\ref{th:v-from-weak-limit}).
We then specialize the abstract procedure to obtain the ZMP method both on bounded (Theorem~\ref{thm:ZMP}), and on unbounded domains (Theorem~\ref{thm:ZMPyuk}). In Section~\ref{sec:num}, we describe some of the numerical experiments.
Section~\ref{sec:summary} summarizes the established results and gives an outlook to possible future work. 
The proofs of our results may be found in Section~\ref{sec:proofs}.
\vfill\null
\pagebreak

\section{Preliminaries} 

\subsection{The Zhao--Morrison--Parr method}\label{zmporigsec} 

The main idea of the ZMP method \cite{ZP1993,ZMP1994} is to rephrase the pointwise constraint $\rho(\rr)=\rhogs(\rr)$ to the equivalent
$D(\rho-\rhogs)=0$, where $D(\rho)$ denotes the Hartree energy of $\rho$ (see Eq.~\eqref{eq:Hartree-term}), and add it to the one-body energy with a penalty parameter $\lambda>0$. We then arrive at the minimization problem
\begin{equation*}
\begin{aligned}
\inf_{\langle \phi_j^\lambda,\phi_k^\lambda \rangle=\delta_{jk}} \Bigg[ & \frac{1}{2}\sum_{j=1}^N \int |\nabla\phi_j^\lambda(\rr)|^2\,\rmd\rr + \int v_\ext(\rr)\rho^\lambda(\rr)\,\rmd \rr \\
& + D(\rho^\lambda) + \lambda D(\rho^\lambda-\rhogs) \Bigg],
\end{aligned}
\end{equation*}
where we have set $\rho^\lambda(\rr)=\sum_{j=1}^N |\phi_j^\lambda(\rr)|^2$ and $v_\ext$ is the external potential. 
Requiring that the derivative of the above functional with respect to the orbitals $\phi_j^\lambda$ vanishes, we obtain the Kohn--Sham-type equations of the ZMP method \cite{ZMP1994}. These equations are then solved for $\lambda\to\infty$.
We may summarize the method as follows.\\

{\it Suppose that $\rhogs$ is a ground-state density for some interacting system with external potential $v_\ext$. Let $\rho^\lambda(\rr) = \sum_{j=1}^N |\phi_j^\lambda(\rr)|^2$, $\vH^\lambda(\rr) = (\rho^\lambda * |\cdot|^{-1})(\rr)$, and 
\begin{equation}\label{zmpeqs0}
    v^\lambda(\rr) = \lambda \int \frac{\rho^\lambda(\rr') - \rhogs(\rr')}{|\rr-\rr'|}\,\rmd \rr' ,
\end{equation}
where the $\phi_j^\lambda$ satisfy
 \begin{equation}\label{zmpeqs}
    \left[ -\frac 1 2 \nabla^2  
    + v_\ext(\rr) + \vH^\lambda(\rr) +v^{\lambda}(\rr)  \right] \phi_j^\lambda(\rr) = e_j^\lambda \phi_j^\lambda(\rr) .
\end{equation}
Then, formally
\begin{equation*}
\vxc(\rr) = \lim_{\lambda\to\infty} v^\lambda(\rr) \quad\mathit{and}\quad \rhogs(\rr) = \lim_{\lambda\to\infty} \rho^\lambda(\rr).
\end{equation*}
}

In Ref.~\cite{ZMP1994} it was reported that solving Eqs.~\eqref{zmpeqs0}-\eqref{zmpeqs} self-consistently as $\lambda\to\infty$, the procedure converges, although it might not be easy to extrapolate to $\lambda = \infty$.

\begin{remark}\leavevmode
\begin{enumerate}
\item We have included $\int v_\ext(\rr) \rho^\lambda(\rr) \,\rmd \rr$ and $D(\rho^\lambda)$ in the minimization problem such that $v^\lambda$ targets $\vxc$ \emph{only}, as can be seen from Eq.~\eqref{zmpeqs}. The Hartree potential in Eq.~\eqref{zmpeqs} is sometimes substituted for a Fermi--Amaldi term, as for example in Ref.~\cite{ZMP1994}.

\item The choice of $D(\cdot)$ as a penalty term in the minimization problem is rather \emph{ad hoc} at this point, and Ref.~\cite{ZP1993} suggested to use the $L^2$-norm instead. One can imagine that a wide variety of functionals should work and we will later show how the corresponding different ZMP schemes arise from different choices for the basic function spaces.

\item Note that $\lambda$ is here just a parameter introduced to penalize for densities differing too much from $\rhogs$. In our interpretation below, the regularization parameter $\varepsilon$ replaces the penalty parameter $\lambda$.
\end{enumerate}
\end{remark}

\subsection{Moreau--Yosida regularization in DFT}
\label{sec:my}

Differentiability of the exact universal (interacting or non-interacting) density functionals cannot be guaranteed, they are in fact \emph{everywhere} discontinuous in the standard formulation~\cite{Lammert2007}. A way to achieve differentiability of \emph{convex} functionals, introduced in the context of DFT in Ref.~\cite{Kvaal2014} (see Refs.~\cite{kvaal2022moreau,Helgaker2022} for recent reviews), is by Moreau--Yosida regularization. This program was extended by two of the present authors in Refs.~\cite{KSpaper2018,CDFT-paper} to different DFT settings, and was also used to show convergence of a modified KS iteration on finite lattices \cite{penz2019guaranteed,penz2020erratum}. 
In this section, we briefly review the basic properties of Moreau--Yosida regularization and their consequences for DFT.

We start in the very general setting of a Banach space $X$, containing all densities under consideration, that is reflexive and where both $X$ \emph{and} its dual $X^*$ are strictly convex. A normed space is strictly convex if the function $\rho\mapsto \|\rho\|^2$ is strictly convex. All Lebesgue spaces $L^p$, $1<p<\infty$, and in general every Hilbert space verify these assumptions, they are even uniformly convex. 
We note that the non-reflexive spaces $L^1$ and $L^\infty$ are excluded.

In the following, an important role will be played by the duality mapping $J$ that maps elements from $X$ to $X^*$ in a canonical way. It is defined as
\begin{equation*}
    J(\rho) = \{\xi \in X^* \mid \langle \xi,\rho \rangle = \|\rho\|_X^2 = \|\xi\|_{X^*}^2 \}.
\end{equation*}
The duality mapping is always homogeneous and in the given setting it is also single-valued and bijective \cite[Prop.~1.117]{Barbu-Precupanu}, so $J^{-1}: X^* \to X$ is well-defined.

The relevance of convex analysis to DFT was recognized very early \cite{Lieb1983}, and so
a series of definitions from this very useful field are in order. We refer the reader to various textbooks \cite{Barbu-Precupanu,bauschke2011convex,zalinescu2002convex} for more information. As mentioned above, exact density functionals
are (typically) not differentiable, however, a more general concept of differentiation is applicable.
The \emph{subdifferential} $\subdiff f : X\to \mathcal{P}(X^*)$ (mapping to the power set of $X^*$) of a convex functional $f:X\to\R\cup\{+\infty\}$ at $\rho\in X$ is defined as the set
\begin{equation*}
\subdiff f(\rho)=\{v \in X^* \mid \forall \sigma \in X : \langle v,\sigma-\rho \rangle \leq f(\sigma)-f(\rho) \}.
\end{equation*}
Intuitively, the convex set $\subdiff f(\rho)$ collects all the supporting hyperplanes of the graph of $f$ at the point
$(\rho,f(\rho))$. If $f$ happens to be continuously differentiable at a point $\rho$, then $\subdiff f(\rho)$ is the singleton set $\{f'(\rho)\}$.  If $\subdiff f(\rho)\neq\emptyset$, then we say that $f$ is \emph{subdifferentiable at $\rho$}.
Analogously, the notion of the \emph{superdifferential} $\superdiff f$ can be introduced for concave functions $f$ via
$\superdiff f=-\subdiff( -f)$. Similarly, whenever $\superdiff f(\rho)\neq\emptyset$, then we say that $f$ is \emph{superdifferentiable at $\rho$}.

The \emph{Legendre transform} $f^*:X^*\to\R\cup\{+\infty\}$ of a functional $f:X\to\R\cup\{+\infty\}$ is defined as
$f^*(v)=\sup\{ \langle v, \rho \rangle - f(\rho) \mid \rho\in X \}$. Similarly, the Legendre transform $g^* : X \to\R\cup\{+\infty\}$ of a functional 
$g:X^*\to\R\cup\{+\infty\}$ is defined as $g^*(v)=\sup\{ \langle v, \rho \rangle - g(v) \mid v\in X^* \}$.
The functionals $f^*$ and $g^*$ are always convex (even if $f$ and $g$ are not). 
The famous \emph{biconjugation theorem} says that whenever $f$ is proper, lower semicontinuous and convex, 
then $f^*$ is proper, (weak-$*$) lower semicontinuous and convex, and furthermore $(f^*)^*=f$.
This result is very important for the convex formulation of DFT.

The \emph{Moreau--Yosida regularization} $f^\epsilon:X\to\R$ with a fixed $\eps > 0$ of a convex, lower semicontinuous functional $f:X \to \R \cup \{+\infty\}$ is given by the lower envelope of the parabola $\rho \mapsto \frac{1}{2\eps} \|\rho\|_X^2$ tracing along the graph $(\rho,f(\rho))$. In a formula this means
\begin{equation}\label{eq:MY-def}
    f^\eps(\rho) = \inf_{\sigma\in X} \left[ f(\sigma)+\frac{1}{2\eps} \|\rho-\sigma\|_X^2 \right].
\end{equation}
Since the parabola is strictly convex and $X$ is reflexive, the above infimum is attained at a unique point. Consequently, the following definition of the \emph{proximal mapping} $\Pi_f^\eps : X\to X$ makes sense,
\begin{equation*}
    \Pi_f^\eps(\rho) = \argmin_{\sigma\in X} \left[ f(\sigma)+\frac{1}{2\eps} \|\rho-\sigma\|_X^2 \right].
\end{equation*}
Both the regularization of a functional and the proximal mapping are exemplified in Fig.~\ref{fig:MY}.

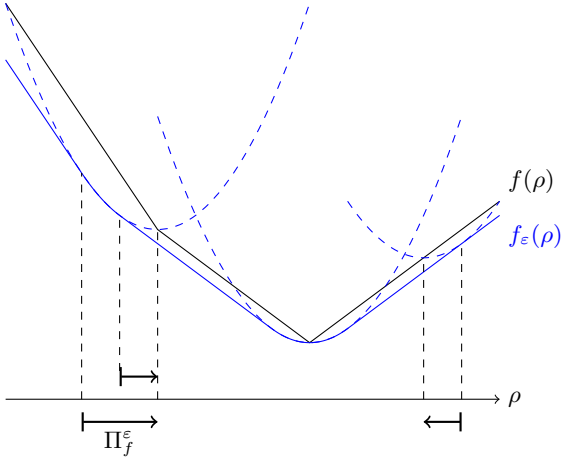
\begin{figure}[ht]
    \centering
    \begin{tikzpicture}[yscale=1.5]
        \draw[->] (-4, 0) -- (2.5, 0) node[right] {$\rho$};
        \draw[-] (2.5, 1.75) node[above right]{$f(\rho)$} -- (0, 0.5) -- (-2, 1.5) -- (-4, 3.5);
        \draw[domain=-2:2, smooth, dashed, variable=\x, blue] plot ({\x}, {1/2*\x*\x+0.5});
        \draw[domain=-0.5:0.5, smooth, variable=\x, blue] plot ({\x}, {1/2*\x*\x+0.5});
        \draw[domain=-4:0, smooth, dashed, variable=\x, blue] plot ({\x}, {1/2*(\x+2)*(\x+2)+1.5});
        \draw[domain=-3:-2.5, smooth, variable=\x, blue] plot ({\x}, {1/2*(\x+2)*(\x+2)+1.5});
        \draw[domain=0.5:2.5, smooth, dashed, variable=\x, blue] plot ({\x}, {1/2*(\x-1.5)*(\x-1.5)+1.25});
        \draw[-, blue] (-0.5, 0.625) -- (-2.5, 1.625);
        \draw[-, blue] (0.5, 0.625) -- (2.5, 1.625)  node[below right]{$f_\eps(\rho)$} ;
        \draw[-, blue] (-3, 2) -- (-4, 3);
        \draw[-, dashed] (-2, 0) -- (-2, 1.5);
        \draw[-, dashed] (-3, 0) -- (-3, 2);
        \draw[|->, thick] (-3, -0.2) -- node[below] {$\proxf$} (-2, -0.2);
        \draw[-, dashed] (-2.5, 0.25) -- (-2.5, 1.625);
        \draw[|->, thick] (-2.5, 0.2) -- (-2, 0.2);
        \draw[-, dashed] (1.5, 0) -- (1.5, 1.25);
        \draw[-, dashed] (2, 0) -- (2, 1.375);
        \draw[|->, thick] (2, -0.2) -- (1.5, -0.2);
    \end{tikzpicture}
    \caption{Visualization of the Moreau--Yosida regularization of a function $f:\R\to\R$ with regularization parabolas and three proximal mappings displayed.}
    \label{fig:MY}
\end{figure}
We now collect all the relevant properties of the Moreau--Yosida regularization, the proofs of which may be found in Refs.~\cite{Kvaal2014,KSpaper2018,Barbu-Precupanu,zalinescu2002convex}.

\begin{theorem}\label{regprop}
Let $X$ be a strictly convex, reflexive Banach space such that $X^*$ is also strictly convex. Let $f:X\to\R \cup \{+\infty\}$ be a convex, lower semicontinuous 
functional. Then the following properties hold true for the Moreau--Yosida-regularized functional $f^\eps : X\to\R$.

\begin{enumerate}
\item\label{MY-property-smoothing} (Smoothing) $f^\eps$ is convex, G\^ateaux differentiable (in particular continuous) and $f^\eps$ is finite everywhere in $X$.
\item (Domination) $\inf f\le f^\eps(\rho) \le f^\delta(\rho) \le f(\rho)$ for all $\rho\in X$ and $0\le \delta<\eps$. In particular, $\inf f^\eps=\inf f^\delta=\inf f$.
\item\label{MY-property-pw-conv} (Pointwise convergence) $f^\eps(\rho)\nearrow f(\rho)$ as $\eps\to 0$ for every $\rho\in X$.
\item\label{MY-property-subdiff} The subdifferential $\subdiff f^\eps : X\to \mathcal{P}(X^*)$ of $f^\eps$ is the singleton $\subdiff f^\eps(\rho)=\{ (f^\eps)'(\rho) \}$,
where $(f^\eps)'(\rho)$ is the G\^ateaux derivative of $f^\eps$ at $\rho\in X$.
\item The proximal mapping $\Pi_f^\eps$ is always single-valued and obeys the limit $\Pi_{f}^\eps(\rho)\to\rho$ as $\eps\to 0$.
\item\label{MY-property-Moreau} $(f^\eps)'(\rho)=\frac{1}{\eps}J(\rho-\Pi_{f}^\eps(\rho))$.
\item\label{MY-property-subdiff2} For any $\rho,\rho_\eps\in X$ the relation $\rho_\eps=\Pi_f^\eps(\rho)$ 
is equivalent to $\frac{1}{\eps}J(\rho-\rho_\eps)\in \subdiff f(\rho_\eps)$.
\item\label{MY-property-Legendre} (Legendre transform) The Legendre transform $(f^\epsilon)^*:X^*\to\R$ of $f^\epsilon$ is given by 
$(f^\epsilon)^*(v)=f^*(v) + \frac{\epsilon}{2}\|v\|_{X^*}^2$ for all $v\in X^*$.
\end{enumerate}
\end{theorem}

Let us briefly discuss the consequences of these properties on a concrete example furnished by the convex and (weakly) lower semicontinuous \cite{Lieb1983}, but the everywhere discontinuous \cite{Lammert2007} (hence nowhere differentiable) \emph{Lieb functional} $F:X\to\R \cup \{+\infty\}$,
\begin{equation}\label{eq:F}
F(\rho) = \inf_{\Gamma \mapsto \rho} \Tr ((K+W)\Gamma), \qquad
W = \frac{1}{2}\sum_{j\neq k}|\rr_j - \rr_k|^{-1},
\end{equation}
on either the Hilbert space $X=L^2(\Lambda)\simeq X^*$, where $\Lambda\subset\R^3$ is a finite box \cite{Kvaal2014}, or, $X=L^3(\R^3)$, $X^*=L^{3/2}(\R^3)$. The Lieb functional may also be defined through the \emph{concave} ground-state energy $E:X^*\to\R$ via
Legendre transformation as
\begin{equation}\label{eevarlieb}
F(\rho)=(-E(-\cdot))^*(\rho)=\sup_{v\in X^*} \Big[ E(v) - \langle v, \rho \rangle \Big],
\end{equation}
which is sometimes referred to as the \emph{Lieb variational principle} \cite{Helgaker2022}.
As usual, the dual pairing between potential and density means 
$\langle v,\rho \rangle = \int v(\rr)\rho(\rr)\rmd\rr$.

Now, consider the Moreau--Yosida regularization $F^\epsilon$ of $F$. 
The \emph{regularized ground-state energy} $E^{\eps}$ (more specifically, the energy corresponding to the regularized functional) given by 
$-E^{\eps}(-\cdot)=(F^\eps)^*$ can then be expressed using property \eqref{MY-property-Legendre} above simply as
\begin{equation}\label{L2regene}
E^{\eps}(v)=E(v) - \frac{\eps}{2} \|v\|_{X^*}^2.
\end{equation}
The one-parameter family of convex and (G\^ateaux-) differentiable functionals $\{F^\eps\}_{\eps>0}$ converges pointwise and increasingly 
to the original, unregularized $F$ from below as $\eps\to 0$ by property \eqref{MY-property-pw-conv} above. The regularized ground-state energy $E^\eps(v)$ also 
converges to the true ground-state energy $E(v)$ according to Eq.~\eqref{L2regene}, and $E^\eps(v)$ is shifted from $E(v)$
by a calculable constant. Further, 
\begin{equation*}
\superdiff E^\eps(v)=\superdiff E(v) + \eps J(v),
\end{equation*}
since $J$ is just the derivative of the concave functional $-\frac{1}{2}\|v\|^2_{X^*}$.
This also means that the degeneracy of the ground state, manifested through multiple ground-state densities $\superdiff E(v)$, is retained exactly for the regularized case. These facts reflect the \emph{``lossless''}
character of the Moreau--Yosida-regularization, i.e., no information is lost in the regularization process. Note, however, that the Moreau--Yosida regularization of $F^\eps$ does \emph{not}
make $E^\eps$ differentiable. Nevertheless, it makes $E^\eps$ superdifferentiable everywhere, while in the unregularized case $E$ is superdifferentiable only at potentials $v\in X^*$ that support a ground state (a.k.a.\ binding potentials). In such a case,
$\superdiff E(v)$ is the set of densities coming from ensemble ground states associated to $v$.

By the biconjugation theorem, the well-known duality relation between $F$ and $E$~\cite{Lieb1983}, sometimes called the \emph{Hohenberg--Kohn variational principle},
\begin{equation}\label{eevarur}
E(v)=\inf_{\rho\in X} \Big[ F (\rho) + \langle v, \rho \rangle \Big],
\end{equation}
also holds in the regularized setting,
\begin{equation}\label{eevar}
E^\eps(v)=\inf_{\rho\in X} \Big[ F^\eps (\rho) + \langle v, \rho \rangle  \Big].
\end{equation}
According to Theorem~\ref{regprop}~\eqref{MY-property-smoothing}, $F^\eps$ is finite everywhere (i.e., even at densities which are not representable by a wavefunction), 
so the minimizer $\rho$ in Eq.~\eqref{eevar} might be unphysical: $\rho$  can be negative, might not integrate to $N$, or 
$\int |\nabla\sqrt{\rho}|^2 \, \rmd \rr$ might diverge (see \cite{Lieb1983}). This is in contrast to minimizers of Eq.~\eqref{eevarur}, which are of course always physical.

The way out from this apparent difficulty is the following.
We know from general convex-analysis considerations~\cite[Prop.~2.33]{Barbu-Precupanu} that $\rho$ is an optimizer in Eq.~\eqref{eevar} if and only if 
$-v\in\subdiff F^\eps(\rho)$, in other words (Theorem~\ref{regprop}~\eqref{MY-property-subdiff}), if and only if
\begin{equation*}
-v=(F^\eps)'(\rho).
\end{equation*}
Combining this with Theorem~\ref{regprop}~\eqref{MY-property-Moreau}, this is equivalent to 
(recall that we use the notation $\rho_\eps=\Pi_{F}^\eps(\rho)$)
\begin{equation*}
-v=\frac{1}{\eps}J(\rho-\rho_\eps).
\end{equation*}
Inverting this equation, we find that a minimizer $\rho$ of Eq.~\eqref{eevar} is always of the form 
\begin{equation}\label{proxdens}
\rho=\rho_\eps - \eps J^{-1}(v).
\end{equation}
We call $\rho_\eps$ the \emph{proximal density} of $\rho$. It is important to note that $\rho_\eps$ will always be physical:
$\rho_\eps\ge 0$, $\int\rho_\eps \,\rmd\rr=N$ and $\int |\nabla\sqrt{\rho_\eps}|^2 \,\rmd\rr<\infty$. This is because $\rho_\eps$ is indeed a minimizer of Eq.~\eqref{eevarur} by Theorem~\ref{regprop}~\eqref{MY-property-subdiff2}.
Again, we see the ``lossless'' character of Moreau--Yosida regularization: from a possibly unphysical minimizer $\rho$ (a.k.a.\ a \emph{quasidensity}) of Eq.~\eqref{eevar}, we can always reconstruct a unique physical density $\rho_\eps=\rho + \eps J^{-1}(v)$.
We can conclude that solving the original minimization problem Eq.~\eqref{eevarur} is completely equivalent to solving a
regularized problem Eq.~\eqref{eevar}. One simply needs to combine Eq.~\eqref{L2regene} and Eq.~\eqref{proxdens} to obtain the ground-state energy $E(v)$ and the respective ground-state density.

Similar considerations hold true for the \emph{(mixed state) kinetic-energy functional}
\begin{equation}\label{eq:T-func}
T(\rho)=\inf_{\Gamma \mapsto \rho} \Tr( K\Gamma),
\end{equation}
which is also convex and lower semicontinuous \cite{Lieb1983}, but discontinuous everywhere.
One obtains the regularized functional $T^\eps$ with equivalent properties as before. However, it is important to note at this point that the functional
\begin{equation*}
T_{\mathrm{S}}(\rho)=\inf_{\sum_{j=1}^N |\phi_j(\rr)|^2=\rho(\rr)\atop \langle \phi_j,\phi_k \rangle=\delta_{jk}} \frac{1}{2} \sum_{j=1}^N \int |\nabla\phi_j|^2\,\rmd\rr
\end{equation*}
used in conventional KS-DFT is \emph{not} convex \cite{Lieb1983}, so that the Moreau--Yosida regularization (or the convex-analysis approach) is not directly applicable in this case.
Therefore, Moreau--Yosida regularized KS methods will always involve fractional occupation numbers since one considers density matrices in Eq.~\eqref{eq:T-func}.
We note in passing, that in the case $\rho$ is 
(non-interacting) $v$-representable with a non-degenerate ground state, then $T(\rho) = T_\mathrm{S}(\rho)$. 

We are now ready to recognize that the ZMP minimization problem initially posed in Section~\ref{zmporigsec} really has the form of a Moreau--Yosida regularization, where the regularization parabola $\frac{1}{2\eps} \|\rho-\sigma\|_X^2$ in Eq.~\eqref{eq:MY-def} is given by $\lambda D(\rho-\rhogs)$, albeit with a non-convex functional $T_{\mathrm{S}} + v_\ext + D$. In order to show that this is not just a accidental similarity and to be able to really formulate the ZMP method rigorously in terms of a Moreau--Yosida regularization, we have to find the suitable setting with a \emph{convex} functional.

\subsection{Model setting}
\label{sec:M-setting}

In KS-DFT the universal density functional (introduced above in Eq.~\eqref{eq:F})
\[
F(\rho) = \inf_{\Gamma \mapsto \rho}\Tr ((K+W)\Gamma) , 
\]
is split into
\begin{equation}\label{eq:F-split}
F(\rho) = T(\rho) + D(\rho) + E_\mathrm{xc}(\rho),
\end{equation}
where
\begin{equation}\label{eq:Hartree-term}
   \DHe(\rho) = \DHe(\rho,\rho), \quad \DHe(\rho,\sigma) = \frac 1 2 \int\!\!\!\int \frac{\rho(\rr)\sigma(\rr')}{|\rr -\rr'|}\,\rmd\rr \,\rmd \rr',
\end{equation}
are the Hartree energy and Coulomb inner product.
The crucial object in Eq.~\eqref{eq:F-split} is the exchange-correlation energy $E_\mathrm{xc}$ that accounts for the non-classical contributions and that actually gets defined through Eq.~\eqref{eq:F-split}. $E_\mathrm{xc}$ has been referred to as ``nature's glue''~\cite{KurthPerdew2000}.

To set the stage, we will relate ground-state densities and their corresponding potentials through convex energy functionals.
First, we introduce the density functional $\Fcnl(\rho)$, given below as a constrained-search functional over density matrices, $\Fcnl_\CS(\rho)$, and an additional term, a convex (and lower semicontinous) functional $\Fcnlp(\rho)>-\infty$, i.e.,
\begin{equation}\label{eq:F-functional}
    \Fcnl(\rho)  =\Fcnl_\CS(\rho)  + \Fcnlp(\rho), \qquad  \Fcnl_\CS(\rho) = \inf_{\Gamma \mapsto \rho} \Tr (\Href\Gamma) .
\end{equation}
The flexibility provided by $\Fcnlp(\rho)$ is discussed in Remark~\ref{remark-setting}.
The functional $\Fcnl(\rho)$ is defined on the Banach space $X$ of densities and for $\Href \neq 0$ is set to $\Fcnl(\rho)=+\infty$ whenever $\rho$ does not represent a physical state.

\begin{remark}\label{remark-setting}\leavevmode
\begin{enumerate}

\item The functional $\Fcnlp$ will typically be non-linear and differentiable.

\item If $\Href$ contains only ``internal'' parts and further $\Fcnlp(\rho)=0$ then $\Fcnl(\rho)$ is indeed a universal density functional, like $\Fcnl(\rho) = T(\rho)$ or $\Fcnl(\rho) = F(\rho)$.

\item\label{remark-setting-item-KS} Relevant for the forthcoming discussion is the choice $\Href = K$ together with $\Fcnlp(\rho) = \langle v_\ext, \rho \rangle + \DHe(\rho)$, so that the potential from the inversion procedure is just the exchange-correlation potential. This will be further explained in Section~\ref{sec:abstract-dens-pot-inv} below.

\item Another possibility is to follow Ref.~\cite{kumar2019universal} and set $\Href = 0$ and for $\Fcnlp(\rho)$ use, e.g., one 
of the functionals
$$
D(\rho), \qquad \int \rho^{1 + \alpha}\,\rmd\rr,  \qquad \frac 1 2\int |\nabla \sqrt{\rho}|^2\,\rmd\rr,
$$
where $\alpha>0$.
These functionals are all convex (and lower semicontinuous).
\end{enumerate}
\end{remark}

The density functional $\Fcnl$ can then be used to compute the (reference) ground-state energy $\EE(v)$ of a system described by $\Fcnl$ ($\Href$ and $\Fcnlp$) with an additional potential $v$,
\begin{equation}\label{eq:E-functional}
    \EE(v) = \inf_{\rho\in X} \Big[ \Fcnl(\rho) + \langle v,\rho \rangle \Big] .
\end{equation}
 Apart from a difference in sign, $\EE(v)$ is the convex conjugate of $\Fcnl(\rho)$, and as such, a concave function. We can also revert back from $\EE(v)$ to $\Fcnl(\rho)$ by duality as
\begin{equation}\label{eq:F-convex-conjugate}
    \Fcnl(\rho) = \sup_{v\in X^*} \Big[ \EE(v) - \langle v,\rho \rangle \Big],
\end{equation}
just like in Eq.~\eqref{eevarlieb}.
The condition for a density being a minimizer in Eq.~\eqref{eq:E-functional} can be equivalently rephrased as (see Section~\ref{sec:my}) 
\begin{equation}\label{eq:subgradient-condition-F}
    -v \in \subdiff \Fcnl(\rho).
\end{equation}
This equation already constitutes a density-potential map and will serve as the starting point for the development of a more practical method. The inversion of this relation, mapping from potentials to densities, entails solving the corresponding KS-type (SCF-) equation with potential $v$ and determining the ground-state density (or densities, in case of degeneracy). By the reciprocity theorem \cite[Prop.~2.33]{Barbu-Precupanu}, this amounts to the superdifferential of the convex conjugate functional $\EE$,
\begin{equation*}
    \rho \in \superdiff \EE(v),
\end{equation*}
which is just the condition for $v$ being a maximizer in Eq.~\eqref{eq:F-convex-conjugate}.
In order to establish a relation to the notion of \emph{$v$-representability}, we consider the case when $\Fcnlp$ is differentiable, such that
\begin{equation}\label{eq:subdiff-split}
    \subdiff \Fcnl = \subdiff \Fcnl_\CS + \Fcnlp'.
\end{equation}
For a potential $-v \in \subdiff \Fcnl$ the difference to $\subdiff \Fcnl_\CS$ just appears as a fixed shift $+ \Fcnlp'$ and the cardinality of the set $\subdiff \Fcnl_\CS$ is the same for $\subdiff\Fcnl$.
This allows to discern three qualitatively different cases when it comes to solutions of Eq.~\eqref{eq:subgradient-condition-F}:

\begin{enumerate}
\item The subdifferential is empty when $\rho$ is not $v$-representable with some Hamiltonian, 
which in this case means $v$-representability with Hamiltonian $\Href+V$.

\item The subdifferential only contains a single element if $\rho$ is $v$-representable and if we are in a setting where the Hohenberg--Kohn theorem holds \cite{reviewPartI}.

\item Finally, the subdifferential can also contain multiple elements in cases where the density is \emph{non-uniquely} $v$-representable, a situation recently discovered in finite lattice system with ground-state degeneracy \cite{penz2021graphDFT}. 
\end{enumerate}

Now, returning to the general case, the situation would be much simpler if the functional $\Fcnl(\rho)$ could be assumed to be functionally (i.e. G\^ateaux-) differentiable, since then the subdifferential always contains exactly one element (the functional derivative). As we discussed above, this is typically not the case, and in particular not the case for exact functionals. 
To remedy the situation, we recall that the desired differentiability do hold for the regularized functionals.
It will be shown below that by appropriately combining these ideas, Moreau--Yosida regularization serves as a basis for the rigorous reformulation of ZMP and related methods.

\section{Main results}
\label{mainsec}

\subsection{Abstract density-potential inversion}\label{sec:abstract-dens-pot-inv}

The main concrete example for the forthcoming discussion is the \emph{(fractional occupation number) Kohn--Sham} scheme, which we choose to model with the functional
\begin{equation*}
\Fcnl_\KS(\rho) = T(\rho) + \langle v_\ext, \rho\rangle + \DHe(\rho)
\end{equation*}
that was already mentioned in Remark~\ref{remark-setting}~\eqref{remark-setting-item-KS}.
The functional $\Fcnl_\KS$ is convex and lower semicontinuous for every $v_\ext\in X^*$. In the KS setting, $\rho\in X$ is \emph{noninteracting} $v$-representable if and only if $\subdiff\Fcnl_\KS(\rho)\neq\emptyset$. 
Note that we have chosen the functional $\Fcnl_\KS$ to be such that only the exchange-correlation energy is left out as the remaining, unknown contribution. 
Then, comparing with Eq.~\eqref{eq:subdiff-split} above, the set of KS potentials reads
\begin{equation*}
    \mathcal{V}_\KS(\rho) := 
    \{ v_\ext + \vH^\rho + \vxc^\rho \mid \vxc^\rho \in -\subdiff\Fcnl_\KS(\rho)\} .
\end{equation*}
This is because $( \langle v_\ext, \cdot\rangle)' \equiv v_\ext$ and $D'(\rho)$ is just the Hartree potential.
Using this language, the density-potential inversion consists of finding a representative of $-\subdiff\Fcnl_\KS(\rhogs)$, where $\rhogs$ is a ground-state of the interacting problem. The following fundamental result says that it is always possible to find this potential as the limit of the functional derivative $-(\Fcnl_\KS^\epsilon)'(\rho)$ of the \emph{regularized functional} $\Fcnl_\KS^\epsilon(\rho)$ as $\epsilon\to 0$. We state the result in its full generality, since we want to apply it in different situations, i.e., not just for $\Fcnl_\KS$.

\begin{theorem}\label{th:v-from-weak-limit}
Suppose that $X$ is a strictly convex Banach space and $X^*$ a uniformly convex one.
Let $\Fcnl:X \to \R \cup \{+\infty\}$ be a convex, lower semicontinuous functional. 
Let $\rhogs \in X$ be such that the subdifferential is nonempty, $\subdiff \Fcnl(\rhogs)\neq\emptyset$. Setting $\rho^\epsilon:=\Pi_\Fcnl^\epsilon(\rhogs)$ as before, we have that the (strong) limit of the sequence
\begin{equation*}
     -(\Fcnl^\eps)'(\rhogs) = \frac{1}{\epsilon} J(\rho^\epsilon - \rhogs) \in -\subdiff \Fcnl(\rho^\epsilon) \quad (\mathrm{as}\;\eps \to 0)
\end{equation*}
in $X^*$ is the unique element $v\in -\subdiff \Fcnl(\rhogs) \subset X^*$ \emph{with minimal norm}.
\end{theorem}

The proof can be found in an even more general form in Section~\ref{sec:proofs}.
Recall that $(\Fcnl^\eps)'(\rhogs)$ always exists due to to regularization. 
To implement the procedure suggested by the theorem, one needs to determine the proximal density $\rho^\epsilon$ and $v^\epsilon:=\frac{1}{\epsilon} J(\rho^\epsilon - \rhogs)\in -\subdiff \Fcnl(\rho^\epsilon)$, and let $\epsilon\to 0$. Of course, in general, neither step is trivial.

Here, the relation $v^\epsilon\in-\subdiff \Fcnl(\rho^\epsilon)$ is equivalent to $\rho^\epsilon\in \superdiff E_\Fcnl(v^\epsilon)$ by the reciprocity relation, where $E_\Fcnl(v)=\inf_\rho[\Fcnl(\rho) + \dua{v}{\rho}]$. This step avoids the necessity to calculate the proximal mapping. For instance, if $\Fcnl=\Fcnl_\KS$, then solving $\rho^\epsilon\in \superdiff E_{\Fcnl_\KS}(v^\epsilon)$ for $\rho^\epsilon$ amounts to solving the (fractional occupation number) Kohn--Sham equations with $v^\epsilon$ in place of the exchange-correlation potential. We return to this case in more detail below.

All this suggests the following \emph{abstract self-consistent scheme}. Suppose that besides the ground-state density $\rhogs\in X$, the initial $v_0^\epsilon$ and $\rho_0^\epsilon$ are given. In a general step we compute $\rho_{i}^\epsilon \in \superdiff E_\Fcnl(v_i^\epsilon)$ and then update the potential via
\begin{equation}\label{potupdate}
v_{i+1}^\epsilon:=\frac{1}{\epsilon} J(\rho_{i}^\epsilon - \rhogs).
\end{equation}
When sufficiently converged, this is repeated for smaller values of $\eps$ to facilitate the limit $\eps\to 0$ in the form of an extrapolation.

We summarize the above considerations in a form of an algorithm.
We use mixing to facilitate convergence, like in the optimal-damping algorithm for electronic-structure calculations \cite{cances2000ODA}.

\begin{procedure}[abstract density-potential inversion algorithm]\label{proc:ZMP}
Fix a strictly decreasing sequence $\eps_k \to 0$ and a mixing parameter $0< \mu <1$.
For a given $v$-representable $\rhogs \in X$, choose corresponding initial values $v_0^k$.
In each $i$-iteration determine the ground-state density $\rho_i^k\in \superdiff E_\Fcnl(v_i^k)$ and calculate the succeeding potential as
\begin{equation}\label{eq:procedure-step-ZMP}
    v_{i+1}^k = (1-\mu)v_i^k + \frac{\mu}{\eps_k} J(\rho_i^k - \rhogs).
\end{equation}
Switch to the next $k$-iteration when $v_i^k$ is sufficiently converged to some $v^k$ and finally let $v^k \to v$ as $k\to\infty$.
\end{procedure}

Here, the initial value $v_0^k$ for the effective potential may be chosen to be the zero potential or the result from the previous $k$-iteration.

It must be noted that in general the convergence of the given procedure is \emph{not} assured (see Section~\ref{sec:num} for a numerical test and some remarks about the choices on $\mu$ and the $\eps$-sequence). Especially the occurrence of degeneracies---when $\superdiff E_\Fcnl(v)$ includes multiple elements to choose from---will be problematic. Yet the potentials that produce such degeneracies are expected to be rare in potential space and thus they are unlikely to be hit by an iteration step $v_i^k$.
Procedure~\ref{proc:ZMP} can be simplified by letting the mixing parameter $\mu$ go to zero, while keeping $\alpha:=\mu\eps^{-1}_k$ constant, thereby dropping the extrapolation steps in $k$. This gives an alternative procedure.

\begin{procedure}[simplified abstract density-potential inversion algorithm]\label{proc:DIFF}
Fix a parameter $0<\alpha<1$. For a given $v$-representable $\rhogs \in X$ choose an initial value for the potential $v_0$. In each iteration step determine the ground-state density $\rho_i \in \superdiff E_\Fcnl(v_i)$ and calculate the succeeding potential as
\begin{equation*}
    v_{i+1} = v_i + \alpha J(\rho_i - \rhogs).
\end{equation*}
Stop when $v_i$ is sufficiently converged.
\end{procedure}

Clearly, the choice of function space $X$ determines the duality mapping $J$ that is crucial for concrete realizations and we now explore a few possibilities that will include two variants of the ZMP method as special cases.

\subsection{Interpretation of the Zhao--Morrison--Parr method}
\label{sec:interZMP}

We already presented the original ZMP approach to density-potential inversion in Section~\ref{zmporigsec}. It is now fairly straightforward to interpret the method in a rigorous way as the density-potential inversion following Eq.~\eqref{potupdate}. The two procedures developed for abstract density-potential inversion would then already serve as possible refinements. Additionally, the discussed cases also add correction terms to the usual ZMP method.

\begin{theorem}[ZMP method on bounded domains]\label{thm:ZMP}
Let $\Omega \subset \R^3$ be a bounded domain and set $X=H^{-1}(\Omega)$, $X^*=H_0^1(\Omega)$. Then $J^{-1} = -\Delta_{\mathrm{D}}$, the Dirichlet--Laplace operator, and $J$ means solving Poisson's equation with zero boundary conditions on the domain $\Omega$, which gives just the Hartree potential with a boundary term.
Then the potential update step Eq.~\eqref{potupdate} reads
\begin{equation}\label{zmpupdate}
    v_{i+1}^\eps (\rr) = \frac{1}{4\pi\eps} \int_\Omega \frac{\rho_i^\eps(\rr')-\rhogs(\rr')}{|\rr-\rr'|}\,\rmd\rr' - \frac 1 \eps g_i^{\eps,\Omega}(\rr),
\end{equation}
where
\begin{equation*}
    g_i^{\eps,\Omega}(\rr) = \int_\Omega \phi_\Omega^\rr(\rr')\left(\rho_i^\eps(\rr')-\rhogs(\rr')\right)\,\rmd\rr'.
\end{equation*}
Here, $\Delta \phi_\Omega^\rr=0$ on $\Omega$ and $\phi_\Omega^\rr(\rr')=(4\pi|\rr'-\rr|)^{-1}$ on $\partial\Omega$. For every $\eps > 0$, this can be solved self-consistently with $\rho_i^\eps \in \superdiff E_\Fcnl(v_i^\eps)$ and extrapolated to $\eps\to 0$. For $\Fcnl=\Fcnl_\KS$ and assuming convergence of $v_i^\eps \to v^\eps$ in the limit $i\to \infty$, the potentials $v^\eps$ then give the exchange-correlation potential as its (strong) limit, $\vxc = \lim_{\eps \to 0 } v^\eps$.
\end{theorem}

Furthermore, by setting $\lambda = (4\pi\eps)^{-1}$, the first term in the right-hand side of Eq.~\eqref{zmpupdate} corresponds \emph{exactly} to the ZMP method Eq.~\eqref{zmpeqs0}, while the second term is a correction that depends on the shape of $\Omega$.

For the choice $X^*=H_0^1(\Omega)$, the Sobolev embedding theorem \cite[4.12,~I.C and III]{adams-book} yields $X^* \subset L^6(\Omega)$ and $\Omega$ bounded allows $X^* \subset L^6(\Omega) \subset L^3(\Omega) \subset L^1(\Omega)$. Thus by duality and denseness of $X^*$ in the $L^p$ spaces above, $X$ includes the usual choice $L^1(\Omega)\cap L^3(\Omega)=L^3(\Omega)$ for the space of densities \cite{Lieb1983}.
Further, it is interesting to note that the same choice of potential space $H_0^1(\Omega)$, with the density as an additional weighting function, already occurred in a study of time-dependent DFT~\cite{penz2011domains}.  

The setting $\Omega=\R^3$ is the one usually presented in molecular DFT. 
In order to formulate the ZMP method on a unbounded domain, we may proceed as follows. Let $X=H^{-1}(\R^3)$ and $X^*=H^1(\R^3)$.
Instead of Poisson's equation we have the \emph{screened} Poisson equation $(-\Delta+\gamma^2)v=\rho$ for a fixed choice $\gamma>0$ with solution, sufficient regularity assumed,
\begin{equation*}
v(\rr)=\frac{1}{4\pi} \int_{\R^3} \frac{\rme^{-\gamma |\rr-\rr'|}}{|\rr-\rr'|} \rho(\rr')\,\rmd\rr',
\end{equation*}
which gives $J(\rho) = v$.
The integral kernel $\frac{1}{4\pi} |\rr|^{-1} e^{-\gamma |\rr|}$ appearing here is known as the  Yukawa potential. We can therefore conclude

\begin{theorem}[ZMP method on unbounded domains]\label{thm:ZMPyuk}
Let $X=H^{-1}(\R^3)$ and $X^*=H^1(\R^3)$. Then the potential update step of Eq.~\eqref{potupdate} reads
\begin{equation*}
    v_{i+1}^\eps (\rr) = \frac{1}{4\pi\epsilon} \int_{\R^3} \frac{ \rho_i^\eps(\rr')-\rhogs(\rr') }{|\rr-\rr'|} \rme^{-\gamma |\rr-\rr'|}\, \rmd\rr'.
\end{equation*}
With this we can proceed like in Theorem~\ref{thm:ZMP}.
\end{theorem}

\subsection{Density-potential inversion on $L^p$ spaces}

Before presenting our numerical results in Section~\ref{sec:num}, we wish to discuss the Hilbert-space setting   
$X=L^2(\Omega)$, $X^*=L^2(\Omega)$, $\Omega \subseteq \R^d$. 
This choice  
implies $J=J^{-1}=\id$. Irrespective of $\Omega$ being bounded or not, $X$ does not include the space $L^1\cap L^3$, so the setting is unsuited for dealing with densities in the continuum. 

This changes in a discrete setting that represents a finite lattice, since for $X=\R^M$, $M$ the number of vertices, all $L^p$ norms, $p\geq1$, are equivalent. 
Using the identity map in Procedure~\ref{proc:ZMP}, we get the na\"{i}ve update scheme
\begin{equation*}
    v_{i+1}^k = (1-\mu)v_i^k + \mu\eps_k^{-1} (\rho_i^k - \rhogs),
\end{equation*}
which tells us to choose a positive (repulsive) potential where the current density is larger than the target density and a negative (attractive) one if it is smaller.
Similarly, Procedure~\ref{proc:DIFF} is just
\begin{equation*}
    v_{i+1} = v_i + \alpha (\rho_i - \rhogs).
\end{equation*}
This update scheme already appeared in the reverse-engineering procedure in Ref.~\cite{karlsson2018disorder} on quantum rings that is itself inspired by the density-potential inversion method of Ref.~\cite{vanleeuwen1994exchange}. The same method was used in Ref.~\cite{kadantsev2004variational}, but with an adaptive parameter $\alpha_i$. The discrete $L^2$ Hilbert-space setting will be further discussed with a numerical experiment in Section~\ref{sec:num} where we compare the different methods.

We also briefly comment on the choice $X=L^{3}(\Omega)$, $X^*=L^{3/2}(\Omega)$, $\Omega \subseteq \R^d$.
This example does \emph{not} consist of Hilbert spaces and it is chosen especially in order to include $L^1\cap L^3$ in $X$ \cite{Lieb1983}. It was previously featured in Ref.~\cite{KSpaper2018}, where the regularized Kohn--Sham iteration was generalized to such density spaces. An example of how the duality mapping looks like for the case $L^p([0,1])$ can be found in Ref.~\cite[Prop.~3.14]{chidume-book}, but it is not a simple expression even in just one dimension ($d=1$).

\subsection{Numerical illustration}
\label{sec:num}

In this section, we consider a numerical example
based on a single system
that serves as a basic proof-of-concept of the density-potential inversion method.
The system under consideration for our numerical exploration is a quantum ring with $M=50$ lattice sites, next-neighbor hopping $\tau=-1$, and a single spinless particle. The space choice is $X=X^*=\R^M$ with the standard $L^2$-norm. We fix a periodic target density $\rhogs$ and solve for the corresponding effective potential by three different methods. The first is just applying a standard non-linear optimization algorithm (BFGS, in its \texttt{scipy} implementation) to Eq.~\eqref{eevarlieb}, while the other two are Procedures~\ref{proc:ZMP} and \ref{proc:DIFF}. The parameter choice is $\mu=0.05$, $\alpha=0.5$, and the relatively short $\eps$-sequence $(1, 0.7, 0.4, 0.1)$. (Note that these parameters cannot be considered universal, a different setting requires a recalibration, so a more careful choice including system size and other parameters seems desirable.) Then all three potentials are again put into the Schrödinger equation and we solve for the ground-state density that is then compared to the target density. Since the potential from the BFGS method turns out to be the most accurate one, the other two potentials are compared to this one. We summarize the results in Fig.~\ref{fig:num_test}.
\begin{figure*}[ht]
    \centering
    \includegraphics[width=1\textwidth]{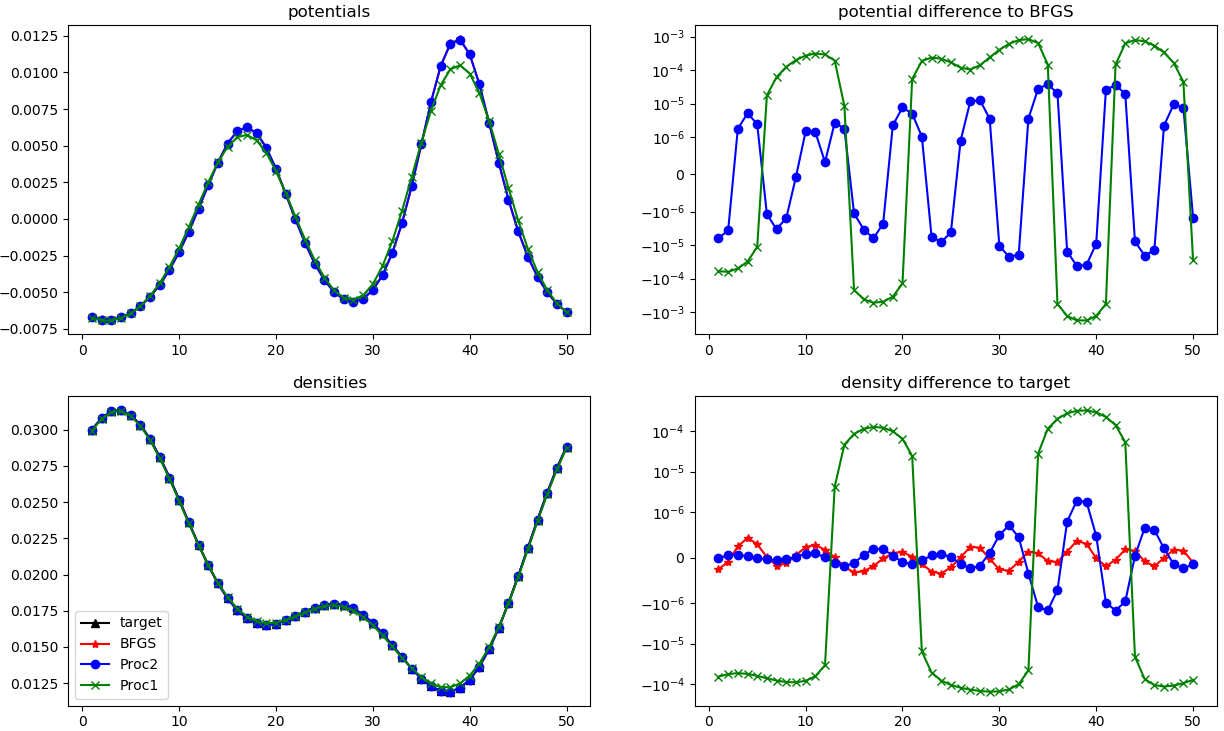}
    \caption{Comparison of different density-potential-inversion methods implemented on a quantum ring with one particle and $M=50$ lattice sites. A target density was given, then different inversion schemes (Procedure~\ref{proc:ZMP}, Procedure~\ref{proc:DIFF} and the standard non-linear optimization algorithm BFGS) were applied to get the potential and to calculate again the ground-state density from there. The left panels show the respective functions, the right panels show the difference to the target density (for densities) and to the BFGS result (for potentials) in a logarithmic scale.}
    \label{fig:num_test}
\end{figure*}
For the chosen tolerance, convergence was achieved in 136 iterations in Procedure~\ref{proc:DIFF} while in every $\eps$-iteration between 43 and 70 steps were required to converge the inner $i$-iteration (with a total of 211 iterations) in Procedure~\ref{proc:ZMP}. This means that within this very specific example, Procedure~\ref{proc:ZMP}, our generalization of the ZMP method, has no benefit compared to the simpler Procedure~\ref{proc:DIFF}, neither regarding iteration speed nor accuracy.
Yet, this finding was only established for the given example and other setups might show a more pronounced difference between the two procedures built upon Moreau--Yosida regularization.
Still, both procedures have comparable overall convergence speed, while the reference optimization algorithm BFGS has a considerably longer runtime (by a factor of 10 in the chosen example).

Note that here the $L^2$-norm was mainly chosen for practicability, since it yields the trivial duality mapping $J=\id$. A different, non-uniform choice of norm could for example focus on specific areas of interest and help to suppress problems with low-density regions.

\section{Summary and future work}
\label{sec:summary}

Starting from the original ZMP method \cite{ZP1993,ZMP1994}, we have formulated a general framework aimed at understanding the density-potential mapping using the Moreau--Yosida regularization scheme applied to DFT. For the convenience of the reader, we reviewed the Moreau--Yosida regularization in detail and summarized its consequences for DFT (Section~\ref{sec:my}).
Next, we formulated our rather general setting that allows to handle a wide variety of density functionals (Section~\ref{sec:M-setting}).
Our main results were stated in Section~\ref{mainsec}, which also contain our reformulation of the ZMP method.
We first proposed an abstract density-potential inversion procedure (Procedure~\ref{proc:ZMP}) by exploiting a property of the derivative of the Moreau--Yosida-regularized density functional (Theorem~\ref{th:v-from-weak-limit}). 
Further, Procedure~\ref{proc:DIFF} yields an even simpler density-potential inversion scheme that was already considered in the literature before \cite{kadantsev2004variational,karlsson2018disorder}. 
We then specialized the abstract procedure to obtain the ZMP method both on bounded domains (Theorem~\ref{thm:ZMP}) and on unbounded domains (Theorem~\ref{thm:ZMPyuk}). In the bounded case, we found that a correction term to the original ZMP method, Eq.~\eqref{zmpeqs0}, is necessary, which accounts for the finite size of the domain. In the unbounded case, the integration kernel of the original ZMP method gets changed to a Yukawa potential. This means that a refinement of the current ZMP inversion algorithms seems entirely possible, e.g., by considering such additional terms or by switching to different, more appropriate spaces for density and potential. Section~\ref{sec:num} finally describes a small proof-of-concept of the different inversion schemes within a numerical experiment.

We stress that the focus of this work is not the immediate development of a new and more efficient practical density-potential inversion method, but to understand the ZMP method and related techniques within a generalized setting based on the Moreau--Yosida regularization. 
In conclusion, our findings provide a framework for devising and analyzing density-potential inversion procedures. It would be also interesting to see how the proposed modifications of the ZMP method perform in practice. Possible modifications are the inclusion of the additional terms derived in Theorems~\ref{thm:ZMP} and \ref{thm:ZMPyuk} or the choice of an altogether different density space $X$, which influences the method through the corresponding duality mapping.

\section{Proofs}\label{sec:proofs}

We first restate Theorem~\ref{th:v-from-weak-limit} in a slightly more general way. The theorem itself is then an immediate consequence if one just remembers that uniform convexity of a space implies it is reflexive and that if a space is reflexive, also its dual space has this property. 

\begin{lemma}
Suppose that $X,X^*$ are both strictly convex and reflexive.
Let $f:X \to \R \cup \{+\infty\}$ be a convex, lower semicontinuous functional and $x \in X$ such that the subdifferential $\subdiff f(x)$ is nonempty. Then there is a unique $\xi \in \subdiff f(x) \subseteq X^*$ with minimal norm that is the weak limit of
\begin{equation}\label{eq:v-from-weak-limit-proofs}
    (f_\eps)'(x) = \eps^{-1} J(x - \proxf(x)) \in \subdiff f(\proxf(x))
\end{equation}
for $\eps \searrow 0$.
Moreover, if $X^*$ is uniformly convex then strong convergence holds.
\end{lemma}

\begin{proof}
Since $f$ is convex and lower semicontinuous, its subdifferential is a maximal monotone operator from $X$ to $X^*$ \cite[Th.~2.43]{Barbu-Precupanu}. Then the first equality and weak convergence follow directly from a result for maximal monotone operators \cite[Prop.~1.146~(iv)]{Barbu-Precupanu}.
We will just add a note about the notation in the reference for the reader's orientation. In Ref.~\cite{Barbu-Precupanu}, the parameter $\eps$ is replaced by $\lambda$, the duality mapping $J$ is called $f$, and the subdifferential $\subdiff f$ is the operator $A$, while the gradients $ (f^\eps)'$ are $A_\lambda$. The connection with regularization is made in Ref.~\cite[Sec.~2.2.3]{Barbu-Precupanu}. That the density is in the domain of the subdifferential means that the subdifferential is not empty at this point. That there exists a unique element with minimal norm in $\subdiff f$ follows from a basic theorem of functional analysis \cite[Cor.~5.1.19]{megginson-book}.

We finally show the set membership in Eq.~\eqref{eq:v-from-weak-limit-proofs}. Taking the definition of the proximal mapping as the minimizer in Eq.~\eqref{eq:MY-def} and translating that into the condition that the subdifferential includes zero at the minimum, we have
\begin{equation*}
    \subdiff f(\proxf(x)) - \eps^{-1} J(x - \proxf(x)) \ni 0.
\end{equation*}
Here we first used that the subdifferential is additive for convex functions and then that $\subdiff \frac{1}{2}\|\cdot\|_X^2 = J$ \cite[Ex.~2.32]{Barbu-Precupanu}.
\end{proof}

\begin{proof}[Proof of Theorem~\ref{thm:ZMP}]
As usual $X^*=H^1_0(\Omega)$ is the Sobolev space with norm $\|v\|_{X^*} = (\|v\|_{L^2}^2 + \|\nabla v\|_{L^2}^2)^{1/2}$ (here $\nabla$ is the weak gradient operator) and zero trace on $\partial\Omega$. The Poincar\'e inequality \cite[6.30]{adams-book} gives an equivalent norm $\|v\|_{X^*} = \|\nabla v\|_{L^2}$ that we will henceforth choose. The dual space $X=H^{-1}(\Omega)$ consists of distributions as explained in Ref.~\cite[3.12]{adams-book}. Now since with the inverse duality mapping and a potential $v\in X^*$ it must hold $\langle v,J^{-1}(v) \rangle = \|v\|^2_{X^*} = \|\nabla v \|_{L^2}^2 =\langle v,-\Delta v\rangle,$ we just get $J^{-1} = -\Delta_{\mathrm{D}}$, the Dirichlet-Laplace operator on $\Omega$. (This also corresponds to the Riesz--Fr\'echet map since we are in the Hilbert space setting.) To obtain $v=J(\rho)$ we must thus solve $-\Delta v = \rho$, $v|_{\partial\Omega}=0$, i.e., Poisson's equation with homogeneous Dirichlet boundary conditions.
If $\rho$ is continuous and we assume a twice continuously-differentiable solution $v$ to exist, then this solution is given by an integral involving the Green function for the domain $\Omega$ \cite[§2.2.4]{evans-book},
\begin{equation*}
    v(\rr)=\int_\Omega \rho(\rr') G(\rr,\rr')\,\rmd\rr'.
\end{equation*}
This Green function can be expressed as $G(\rr,\rr') = \Phi(\rr-\rr')-\phi_\Omega^\rr(\rr')$ with $\Phi(\rr)=1/(4\pi|\rr|)$ the fundamental solution of Laplace's equation and $\phi_\Omega^\rr(\rr')$ the corrector function solving $\Delta'\phi_\Omega^\rr(\rr')=0$ on $\Omega$ with boundary condition $\phi_\Omega^\rr(\rr') = \Phi(\rr'-\rr)$ on $\partial\Omega$. Inserting this into Eq.~\eqref{potupdate} gives the stated formula. The convergence to the exchange-correlation potential is then an application of Theorem~\ref{th:v-from-weak-limit}.
\end{proof}

\begin{proof}[Proof of Theorem~\ref{thm:ZMPyuk}] The only difference between the previous proof is a change in the duality map (because the function space has been changed). The usual norm of $X^*=H^{1}(\R^3)$ is $\|v\|_{X^*} = (\|v\|_{L^2}^2 + \|\nabla v\|_{L^2}^2)^{1/2}$ but for any $\gamma>0$ one has $\|v\|_{X^*} = (\gamma^2\|v\|_{L^2}^2 + \|\nabla v\|_{L^2}^2)^{1/2}$ as an equivalent norm. Then the condition for the duality map is $\langle v,J^{-1}(v) \rangle = \gamma^2\|v\|_{L^2}^2 + \|\nabla v\|_{L^2}^2 = \langle v,\gamma^2 v \rangle + \langle v,-\Delta v \rangle$ and thus $J^{-1} = -\Delta + \gamma^2$. Consequently, application of $J$ means solving the screened Poisson equation.
\end{proof}

\section*{Data Availability Statement}

The numerical example of Section~\ref{sec:num} was implemented in a small Python script that is part of the public domain and can be found at \small\url{https://github.com/ERC-REGAL/REGAL/}.

\section*{Acknowledgements}
AL and MAC acknowledge funding through ERC StG REGAL under agreement No.\ 101041487, AL and MAC also received partial support from the Norwegian Research Council through Grant Nos.\ 287906 (CCerror) and 262695 (CoE Hylleraas).

\section*{References}

%

\end{document}